\documentclass[11pt]{article}
\usepackage{amsmath,amsthm,amssymb}
\newcommand{\remove}[1]{}
\newtheorem{thm}{Theorem}[section]
\newtheorem{claim}[thm]{Claim}
\newtheorem{lem}[thm]{Lemma}

\newtheorem{cor}[thm]{Corollary}

\newtheorem{remark}[thm]{Remark}

\def\F{{\mathbb{F}}}
\def\P{{\mathbb{P}}}
\def\bF{{\overline{\mathbb{F}}}}

\def\P{{\mathbb{P}}}

\def\cO{{\cal O}}
\def\cM{{\cal M}}

\def\V{{\mathbf{V}}}

\newcommand{\Laurent}[1]{{\bF\{\{T\}\}}}

\def\_{\,\,\,\,\,}

\def\modulo{\text{mod}}

\newcommand{\eps}{\epsilon}

\begin{document}

\title{Variety Evasive Sets}
\author{
Zeev Dvir\thanks{Department of Computer Science and Department of Mathematics, Princeton University, Princeton NJ.
Email: \texttt{zeev.dvir@gmail.com}. Research partially
supported by NSF grant CCF-0832797 and by the Packard fellowship.}
\and
J\'anos Koll\'ar\thanks{Department of Mathematics, Princeton University, Princeton NJ.
Email: \texttt{kollar@math.princeton.edu}.}
\and
Shachar Lovett
\thanks{School of Mathematics, Institute for Advanced Study, Princeton, NJ. Email:
\texttt{slovett@math.ias.edu}. Research supported by NSF grant DMS-0835373.}\\
}
\date{}
\maketitle

\begin{abstract}
	We give an explicit construction of a large subset $S \subset \F^n$, where $\F$ is a finite field, that has small intersection with any affine variety of fixed dimension and bounded degree. Our construction generalizes a recent result of Dvir and Lovett (STOC 2012) who considered varieties of degree one (that is, affine subspaces).
\end{abstract}

\section{Introduction}\label{sec-intro}

In this work we consider subsets of $\F^n$, where $\F$ is a finite field. We will be interested in constructing large subsets of $\F^n$ that have small intersection with any $k$-dimensional affine variety of bounded `complexity'. Our measure of complexity here will just be the degree of the variety. We call such sets {\em variety evasive sets}. One can show, using the probabilistic method, that a large random set  will have small intersection (small here means independent of the field size) with any $k$-dimensional variety of bounded degree (see Section~\ref{sec-random} for the probabilistic bound). We give an {\em explicit} construction of such a set and provide  quantitative bounds on the intersection with varieties of sufficiently small degree. By `explicit' here we mean that there is an efficient algorithm that outputs elements in the set, given an index, in a one-to-one manner. 

Our work builds on an earlier work by a subset of the authors \cite{DvirLovett11} in which such a construction was given for varieties of degree one -- affine subspaces. The original motivation for the work done in \cite{DvirLovett11} was an improvement to the list-decoding algorithm of Guruswami-Rudra codes \cite{GuruswamiRudra08,Guruswami11}. We are not aware of any applications of variety evasive sets but hope that these will indeed prove useful in the future.

 Our starting point is a new, more direct, proof of the main theorem of \cite{DvirLovett11}. The new proof technique allows us to  generalize the result to higher degree varieties. The new proof uses a lemma on Laurent series solutions (Lemma~\ref{lem-laurent}) which was implicitly used in earlier works dealing with explicit constructions of graphs with pseudo-random properties \cite{KRS96}.

The main ingredient in our construction is a theorem (Theorem~\ref{thm-variety-basic})  that gives an explicit set of $k$ polynomials $f_1,\ldots,f_k \in \F[x_1,\ldots,x_n]$ such that the variety that they define (over the algebraic closure of $\F$) has zero dimensional intersection with {\em any} $k$-dimensional variety of degree at most $d$. The degrees of these $k$ polynomials depend on both the degree parameter $d$ and the number of variables $n$. The result for finite fields follows by showing that these polynomials have  a large (and easy to describe) set of common solutions over the finite field $\F$.

\paragraph{Organization:} In Section~\ref{sec-algclosure} we work over the algebraic closure of $\F$ and show how to construct the polynomials $f_1,\ldots,f_k$ discussed above. Section~\ref{sec-twolemmas} contains the proof of the main lemma regarding Laurent series solutions as well as another useful lemma on projections of varieties. The two main theorems of Section~\ref{sec-algclosure}, Theorem~\ref{thm-variety-basic} and Theorem~\ref{thm-variety-buckets}, are proved in Sections~\ref{sec-proofofthm} and
\ref{sec-proofbucket}. In Section~\ref{sec-finite} we  go back to the original problem and discuss the zero-set of $f_1,\ldots,f_k$ over the finite field $\F$. In Section~\ref{sec-random} we compare our explicit construction to that obtained by a random construction. Finally, in Section~\ref{sec-conjecture} we discuss some connections between our work and a conjecture of Griffiths and Harris.

\section{Variety Evasive Sets in The Algebraic Closure}\label{sec-algclosure}

Let $\F$ be a field and $\bF$ its algebraic closure. Given $k$ polynomials $f_1,\ldots,f_k \in \bF[x_1,\ldots,x_n]$, we denote the variety they define as
$$
\V(f_1,\ldots,f_k) := \{ x \in \bF^n \,|\, f_1(x) = \ldots = f_k(x) = 0\}.
$$

We will use the following definition: A $k \times n$ matrix (where $k \le n$) is {\em $k$-regular} if all its $k \times k$ minors are regular (i.e have non-zero determinant). For example, if $\F$ is a field with at least $n$ distinct nonzero elements $\gamma_1,\ldots,\gamma_n$ then the Vandermonde matrix
$A_{i,j}=\gamma_j^i$ is $k$-regular.

The following theorem is the heart of our construction and is proved in Section~\ref{sec-proofofthm}.

\begin{thm}\label{thm-variety-basic}
Let $1 \le k \le n$ and $d \geq 1$ be integers  and let $\F$ be a field.
Let $A$ be a $k \times n$ matrix with coefficients in $\F$ which is $k$-regular. Let $d_1>d_2>\ldots>d_n > d$ be pairwise relatively prime integers.
Let the polynomials $f_1,\ldots,f_k \in \F[x_1,\ldots,x_n]$ be defined as follows:
$$
f_i(x_1,\ldots,x_n) := \sum_{j=1}^n A_{i,j} \cdot x_j^{d_j}
$$
Let $U=\V(f_1,\ldots,f_k) \subset \bF^n$ denote the variety defined by these polynomials.
Then, for every affine variety $V \subset \bF^n$ of dimension $k$ and degree at most $d$, the variety $V \cap U$ has dimension zero. In particular,
$$
|V \cap U| \le d \cdot \prod_{i=1}^{k} d_i.
$$
\end{thm}


Choosing $d_1,\ldots,d_n$ in Theorem~\ref{thm-variety-basic} to be the first $n$ prime numbers larger than $d$, we get that $d_1 \le c (d+n) \log(d+n)$ for an absolute constant $c>0$ and hence
\begin{equation}\label{eq:bound-basic}
|V \cap U| \le d \cdot (c (d+n) \log(d+n))^k \le (d+n)^{O(k)}.
\end{equation}

This bound is quite effective when the degree $d$ is comparable to the number of variables $n$. In some scenarios it is useful to obtain better bounds when $d \ll n$. This is achieved by the following construction. Fix $m>k$ such that $m$ divides $n$. Let $U \subset \bF^m$ be the variety constructed by Theorem~\ref{thm-variety-basic} in dimension $m$. We show that the $(n/m)$-Cartesian product $U^{n/m}$ has a bound on its intersection with any variety $V \subset \bF^n$ of dimension $k$ and degree $d$, and this bound depends just on $m$ and not on $n$. Recall that if $U \subset \bF^m$ is a variety then its $(n/m)$-Cartesian product $U^{n/m} \subset \bF^{n}$ is the variety given by
$$
U^{n/m} = \{x \in \bF^{n}: (x_1,\ldots,x_m) \in U, (x_{m+1},\ldots,x_{2m}) \in U, \ldots, (x_{n-m+1},\ldots,x_{n}) \in U \}.
$$

We prove the next theorem in Section~\ref{sec-proofbucket}.
\begin{thm}\label{thm-variety-buckets}
Let $k,d \geq 1$ be integers  and let $\F$ be a field. Let $m>k$ be an integer such that $m$ divides $n$.
Let $A$ be a $k \times m$ matrix with coefficients in $\F$ which is $k$-regular. Let $d_1>d_2>\ldots>d_m > d$ be pairwise relatively prime integers. Let the polynomials $f_1,\ldots,f_k \in \F[x_1,\ldots,x_m]$ be defined as follows:
$$
f_i(x_1,\ldots,x_m) := \sum_{j=1}^m A_{i,j} \cdot x_j^{d_j}
$$
Let $U=\V(f_1,\ldots,f_k) \subset \bF^m$ be the variety defined by these polynomials.
Then, for every affine variety $V \subset \bF^n$ of dimension $k$ and degree at most $d$,
$$
|V \cap U^{n/m}| \le d^{k+1} \cdot (\prod_{i=1}^k d_i)^k.
$$
\end{thm}

In particular, if we fix some $\eps > 0$, set $m=\lfloor k/\eps \rfloor$ and let $d_1,\ldots,d_m$ be the first $k$ primes following $d$ then $U^{n/m}$ has dimension at least $(1-\eps) n$ and
\begin{equation}\label{eq:bound-buckets}
|V \cap U^{n/m}| \le (d+m)^{O(k^2)}.
\end{equation}

We do not know if the bound on $|V \cap U^{n/m}|$ achieved by Theorem~\ref{thm-variety-buckets} can be improved to match that of $|V \cap U|$ established in Theorem~\ref{thm-variety-buckets}. Our current analysis only imply a weaker bound. We note that when $d=1$ it was shown in \cite{DvirLovett11} that in fact $|V \cap U^{n/m}| \le m^k$. We suspect that for general $d$, the bound in Theorem~\ref{thm-variety-buckets} can be improved to $|V \cap U^{n/m}| \le (d+m)^{O(k)}$.

\section{Two Lemmas}\label{sec-twolemmas}

\subsection{A lemma on Laurent series solutions}

 A Laurent series in the variable $T$ is a formal expression of the form
\[  h(T) = \sum_{j = -r}^{\infty}b_j\cdot T^{j}\]
That is, a formal power series in $T$ that has a finite number of negative powers.  The set of formal Laurent series in variable $T$ and coefficients from $\bF$ will be denoted by $\Laurent{T}$. If $f(x_1,\ldots,x_n)$ is a polynomial with coefficients in $\bF$ and $h_1,\ldots,h_n \in \Laurent{T}$, we say that $f$ vanishes on $h_1,\ldots,h_n$ if $f(h_1(T),\ldots,h_n(T))$ is the zero element of $\Laurent{T}$. Notice that, in the evaluated polynomial, each coefficient of $T$ is a function of a finite number of coefficients in the $h_i$'s and so the output is a well defined Laurent series. We say that $h(T)$ has a {\em pole} if there is at least one negative power of $T$ appearing in it with a non-zero coefficient.

The following lemma states that every affine variety of dimension at least one has a solution in Laurent series such that at least one coordinate has a pole. It was originally used in \cite{KRS96} (Remark 1) but was not stated there explicitly. The proof will use basic notions and results from the theory of algebraic curves. All of the results we will use can be found in the first two chapters of \cite{Shaf}.

\begin{lem}\label{lem-laurent}
Let $V \subset \bF^n$ be an affine variety of dimension $k \geq 1$ and let $I(V) \subset \bF[x_1,\ldots,x_n]$ be the ideal of polynomials that vanish on $V$. Then there exists $h_1(T),\ldots,h_n(T) \in \Laurent{T}$ such that, for all $f \in I(V)$, $f$ vanishes on $h_1,\ldots,h_n$. In addition, at least one of the $h_i$'s has a pole.
\end{lem}
\begin{proof}
We follow the argument given in  \cite[Remark 1]{KRS96}. Let $C \subset V$ be an irreducible curve contained in $V$ and let $I(C)$ be its ideal so that $I(V) \subseteq I(C)$. Consider the embedding of $\bF^n$ into projective space $\P\bF^n$ by adding a new coordinate $x_0$ (so that $\bF^n$ is identified with the set $x_0=1$). Let $\bar C$ be the projective closure of $C$ in $\P\bF^n$. Since a curve and a hyperplane always intersect in projective space, we have that $\bar C$ contains a point $P_0$ with $x_0 = 0$ (i.e., a point at infinity). We would like to work with power series solutions at $P_0$ but this is problematic since $P_0$ might be singular. To remedy this, we use the fact that there always exists a non-singular irreducible projective curve $C'$ and a surjective morphism $\phi : C' \mapsto C$. Let $Q_0 \in C'$ be a (non-singular) point such that $\phi(Q_0) = P_0$. Let  $\cO_{Q_0}$ be the local ring of $Q_0$ and $\cM_{Q_0} \subset \cO_{Q_0}$ its maximal ideal. Since $\phi$ is a morphism, its coordinates can be written locally as $n+1$  functions $\phi_0,\ldots,\phi_n \in \cO_{Q_0}$. Since $Q_0$ is non-singular, there is an injective ring homomorphism $\tau$ mapping $\cO_{Q_0}$ to the ring $\bF[[T]]$ of formal power series in $T$ in such a way that $\cM_{Q_0}$ maps into the maximal ideal $I_0 \subset \bF[[T]]$ containing all power series that are divisible by $T$ (i.e., those that have a zero constant term). Define, for each $0 \leq i \leq n+1$ the power series $g_i(T) = \tau(\phi_i)$ corresponding to $\phi_i$.

Notice that, since the $0$'th coordinate of $P_0$ is zero, we have that $\phi_0 \in \cM_{Q_0}$ and so $g_0(T)$ has a zero constant term. Also, since $P_0$ has at least one non zero coordinate, there is some $g_i(T)$ with a non-zero constant term. Consider the formal Laurent series $h_i(T) = g_i(T)/g_0(T)$, where $i \in [n]$.  From the above comments on the constant terms in the $g_i$'s we get that some $h_i$ has a pole. We now show that the $h_i$'s satisfy the consequence of the lemma. Let $f \in I(C)$ and let $\bar f \in \bF[x_0,\ldots,x_n]$ be its homogenization defined as $$ \bar f(x_0,\ldots,x_1) = x_0^{\deg(f)}\cdot f(x_1/x_0,\ldots,x_n/x_0).$$ Since $\bar f$ vanishes on $\bar C$ we have the identity $\bar f(\phi_0, \ldots,\phi_n)=0$ over the ring $\cO_{Q_0}$. This implies that $\bar f(g_0(T), \ldots, g_n(T))=0$ as a formal power series identity. Dividing by  $g_0(T)^{\deg(f)}$ we get that $f(h_1(T),\ldots,h_n(T))=0$. This completes the proof.
\end{proof}

\begin{remark}
	Even if $V$ is defined over $\F$, the coefficients of the Laurent series solutions given by the lemma are generally not going to be in $\F$ but only in the algebraic closure $\bF$.
\end{remark}

\begin{remark}
	The fact that one of the  Laurent series has a pole is what makes this lemma so useful (as we shall see in the proof of Theorem~\ref{thm-variety-basic}). This  pole allows us to work only with the leading terms of the series instead of having to analyze higher order terms (as we would have to do with power series solutions). The basic fact we will use is that, if $h(T)$ has a pole of order $r$ then $h(T)^d$ has a pole of order $rd$ for every $d \geq 1$.
\end{remark}

\subsection{A lemma on projections of varieties}

\begin{lem}\label{lem-projection}
Let $V \subset \bF^n$ be an affine variety of dimension $k < n$ and degree $d$. Then, for every set $J \subset [n]$ of size  $k+1$ there exists a polynomial $f \in \bF[x_j, \,\, j \in J] $ (i.e., a polynomial that depends only on variables indexed by $J$) with degree at most $d$ such that $f \in I(V)$. Moreover, if $V$ is defined over a subfield $\F$ then the coefficients of $f$ can be chosen to be in the same subfield.
\end{lem}
\begin{proof}
We can assume w.l.o.g that $V$ is irreducible (otherwise apply the theorem on each of the components). We now proceed by induction on $n$ (the base case $n=2$ is trivial). If $V$ is a hyper-surface then its degree is equal to the degree of its defining polynomial and so we are done. If $k < n-1$ we take a projection of $V$ on some $n-1$ coordinates containing $J$. The projection  is, in general, not an affine variety. However, one can show that the projection is always an open set of one (that is, an affine variety minus some proper sub-variety) \cite{Shaf}. Consider the Zariski closure of the projection. This is a  variety of dimension at most $k$ and degree at most $d$ (See \cite[Lemma 2]{Heintz83} for an elementary proof and also \cite[Sec IV.1, Ex.5]{Shaf}). Thus, using the inductive hypothesis, there exists a polynomial $f$ with the required properties.

To prove the moreover part, notice that the ideal $I$ of $V$ is generated by a finite number of polynomials $g_1,\ldots,g_t$ with coefficients in $\F$. We have shown that there exists an $f \in \bF[x_j, \,\, j \in J] \cap I$ of degree at most $d$. This means that there exist polynomials $h_1,\ldots,h_t$ such that $f = \sum_i h_i\cdot g_i$. Consider the linear map $H$ sending a $t$-tuple of polynomials $h_1,\ldots,h_t$ to their combinations $\sum_i h_i \cdot g_i$. This linear map is defined over $\F$ since the coefficients of $g_1,\ldots,g_t$ are in $\F$. We know that the image of $H$ contains an element in $\bF[x_j, \,\, j \in J]$ and so, since it is defined over $\F$, it must also contain an element with entries in $\F$. 
\end{proof}

\section{Proof of Theorem~\ref{thm-variety-basic}}\label{sec-proofofthm}

Let $V \subset \bF^n$ be an affine variety of dimension $k$ and degree $d$ and consider the intersection of $V$ and $U = \V(f_1,\ldots,f_k)$. Ie we know that $V \cap U$ has dimension zero, the bound on the size $|V \cap U|$ will follow  from the affine B\'ezout inequality. Specifically, since applying invertible row operations on a the matrix $A$ in the construction does not affect the variety $\V(f_1,\ldots,f_k)$, we can assume w.l.o.g that the first $k$ columns of $A$ form an upper triangular square matrix. Thus  $\deg(f_i) = d_i$ for all $i \in [k]$. Applying the affine B\'ezout's inequality we get that the degree of the intersection $V \cap U$ is at most the product of the degree of $V$ and the degree of $U = \V(f_1,\ldots,f_k)$ which is (again by B\'ezout) a most $d_1\cdot \ldots \cdot d_k$. A variety of dimension zero and degree $D$ has at most $D$ points.

We now turn to  showing that $V \cap U$ has dimension zero. Assume by contradiction that its dimension is at least $1$. Then, by Lemma~\ref{lem-laurent}, there exist Laurent series $h_1(T), \ldots,h_n(T) \in \Laurent{T}$, one of which has a pole, such that
\begin{enumerate}
	\item All polynomials in $I(V)$ vanish on $h_1,\ldots,h_n$,
	\item $f_1,\ldots,f_k$ vanish on $h_1,\ldots,h_n$.
\end{enumerate}
Consider the second item and write the $k$ identities
\begin{eqnarray*}
	\sum_{j=1}^n A_{ij} \cdot h_j(T)^{d_j} = 0, \,\,\, i \in [k].
\end{eqnarray*}
Let $R$ denote the largest integer so that  $T^{-R}$ appears with non zero coefficient in one of the Laurent series $h_j(T)^{d_j}$, $j \in [n]$. Since at least one $h_j$ has a pole we know that $R$ is positive. Since $A$ is regular we conclude that the term $T^{-R}$ has to appear with non zero coefficient in at least $k+1$ of the Laurent series $h_j(T)^{d_j}$, $j \in [n]$. To see why, notice that the minimal (negative) power of $T$ has to cancel in all $k$ equations and so, if there were less than $k+1$ places where $T^{-R}$ appears, we would get a non zero linear combination of at most $k$ columns of $A$ that vanishes, contradicting the regularity of $A$. Let $J \subset [n]$ denote the set of $j$'s such that $h_j(T)^{d_j}$ has a non zero coefficient of $T^{-R}$. For each $j \in J$, let $r_j$ be the largest integer such that $T^{-r_j}$ has a non zero coefficient in $h_j(T)$. From the maximality of $R$ we get that $ R = r_j \cdot d_j$ for all $j \in J$.

Write $J = \{j_1,\ldots,j_{k+1}\}$ (if $J$ is larger than $k+1$ we take some subset of $J$ of this size). We now use item (1) above, namely that $h_1,\ldots,h_n$ satisfy the equations of $V$, together with Lemma~\ref{lem-projection}, to conclude that there exists a polynomial $g(Z_1,\ldots,Z_{k+1})$ in $k+1$ variables and of degree at most $d$ such that
\begin{equation}\label{eq-g(h)}
	g(h_{j_1}(T), \ldots,h_{j_{k+1}}(T)) = 0.
\end{equation}
 Each monomial in $g$ is of the form $Z_1^{\alpha_1}\cdot \ldots, \cdot Z_{k+1}^{\alpha_{k+1}}$ with $\sum_i \alpha_i \leq d$. We identify each monomial with the vector of non-negative integers $\alpha_1,\ldots,\alpha_{k+1}$. Consider the smallest (negative) power of $T$  that appears in one of the monomials of $g$ after the substitution $Z_i = h_{j_i}(T)$. This power of $T$ must appear in at least  two distinct monomials (otherwise it will not cancel). Let $\alpha = (\alpha_1,\ldots,\alpha_{k+1})$ and $\beta = (\beta_1,\ldots,\beta_{k+1})$ be  two such monomials. Thus, we have the equality
\begin{equation}\label{eq-alphabeta}
	\sum_{i \in [k+1]} \alpha_i \cdot r_{j_i} = \sum_{i \in [k+1]} \beta_i \cdot r_{j_i}.
\end{equation}
Let $D = \prod_{i}d_{j_i}$ and observe that $r_{j_i} = R/d_{j_i}$ for each $i \in [k+1]$. Now, multiply Eq.(\ref{eq-alphabeta}) by $D/R$ and obtain the equality
\begin{equation}
	\sum_{i \in [k+1]} \alpha_i \cdot (D/d_{j_i}) = \sum_{i \in [k+1]} \beta_i \cdot (D/d_{j_i}).
\end{equation}
Taking this equality modulo $d_{j_1}$ we get
$$ \alpha_1 \cdot (D/d_{j_1}) = \beta_1 \cdot (D/d_{j_1}) \,\,\modulo  \,\, d_{j_1}. $$
Since $D/d_{j_1}$ is co-prime to $d_{j_1}$ we can cancel it from both sides and get
$$ \alpha_1 = \beta_1\,\, \modulo \,\, d_{j_1}.$$
Now, since both $\alpha_1$ and $\beta_1$ are at most $d < d_{j_1}$ we get the equality $\alpha_1 = \beta_1$. Repeating this argument for all $i \in [k+1]$ we get $\alpha = \beta$ which contradicts our assumption. This completes the proof of Theorem~\ref{thm-variety-basic}.

\section{Proof of Theorem~\ref{thm-variety-buckets}}\label{sec-proofbucket}
Let $m$ divide $n$ and let $U=\V(f_1,\ldots,f_k) \subset \bF^m$ given by Theorem~\ref{thm-variety-buckets}. Let $V \subset \bF^n$ be a variety of dimension $\dim(V) \le k$ and degree $\deg(V) \le d$. We will show that
$$
|V \cap U^{n/m}| \le \deg(V)^{\dim(V)+1} \cdot (\prod_{i=1}^k d_i)^{\dim(V)}.
$$
We will prove the bound by induction on the number of buckets $n/m$. We note that the base case $n=m$ was established (with a better bound) in Theorem~\ref{thm-variety-basic}. Hence, we assume $n \ge 2m$.

Moreover, we note that it suffices to prove the bound when $V$ is irreducible. Otherwise, let $V=V_1 \cup \ldots \cup V_t$ be the decomposition of $V$ to irreducible components. If we establish the bound for each component $V_1,\ldots,V_t$ individually, then since $\sum_{i=1}^t \deg(V_i)=\deg(V) \le d$ and $\dim(V_i) \le \dim(V) \le k$ we obtain that
$$
|V \cap U^{n/m}| \le \sum_{i=1}^t |V_i \cap U^{n/m}| \le \sum_{i=1}^k \deg(V_i)^{k+1} \cdot (\prod_{j=1}^k d_j)^k
\le d^{k+1} \cdot (\prod_{j=1}^k d_j)^k.
$$
Hence, we assume from now on that $V$ is irreducible (however, by the above claim by can apply the result inductively to reducible varieties as well).

Let $\pi(V)$ denote the projection of $V$ to the first $m$ coordinates,
$$
\pi(V)=\{(x_1,\ldots,x_m): x \in V\} \subset \bF^m.
$$
Notice that we already know that $\pi(V \cap U^{n/m})$ is finite (this follows from Theorem~\ref{thm-variety-basic}). Our task is to show that the size is smaller than what you would get with a careless application of that Theorem.
For each $a \in \pi(V)$, let $\varphi(V,a)$ denote the fiber of $V$ over $a$,
$$
\varphi(V,a)=\{(x_{m+1},\ldots,x_n): x \in V, x_1=a_1,\ldots,x_m=a_m\} \subset \bF^{n-m}.
$$
We will apply the identity
\begin{equation}\label{eq:induction_size}
|V \cap U^{n/m}| = \sum_{a \in \pi(V) \cap U} |\varphi(V,a) \cap U^{(n/m)-1}|.
\end{equation}

As in Lemma~\ref{lem-projection}, the projection $\pi(V)$ is in general not an affine variety, but is an open subset of an affine variety. Let $\overline{\pi(V)}$ denote its Zariski closure. We note that $\overline{\pi(V)}$ has degree at most $d$ as discussed in Lemma~\ref{lem-projection}. We further note that the fibers $\varphi(V,a)$ are affine varieties of degree at most $d$ because they are the intersection of $V$ with the degree-one variety given by $x_1=a_1,\ldots,x_m=a_m$.

Consider first that case that $\pi(V)$ is zero dimensional, hence finite. Since we assume $V$ is irreducible we must have that $|\pi(V)|=1$. That is, $V=\{a\} \times \varphi(V,a)$ for some $a \in \bF^m$. The bound on $|V \cap U^{n/m}|$ then follows immediately by induction, since
$$
|V \cap U^{n/m}| \le |\varphi(V,a) \cap U^{n/m-1}|
$$
and $\varphi(V,a)$ has the same dimension and degree as that of $V$.
So, assume $\ell=\dim(\pi(V)) \ge 1$. By Theorem~\ref{thm-variety-basic} we know that
\begin{equation}\label{eq-7}
|\pi(V) \cap U| \le |\overline{\pi(V)} \cap U| \le d \cdot \prod_{i=1}^{k} d_i.
\end{equation}
In fact, one can obtain the improved bound $|\pi(V) \cap U| \le d \cdot \prod_{i=1}^{\ell} d_i$, however this will only obtain a marginal improvement in the overall bound, so we avoid it. Consider a fiber $\varphi(V,a)$ for $a \in \pi(V) \cap U$. We claim that
$$\dim(\varphi(V,a)) \le \dim(V)-1.$$ Otherwise, $\{a\} \times \varphi(V,a)$ is an affine variety contained in $V$ and with the same dimension as that of $V$. Since by assumption $V$ is irreducible this implies that $V=\{a\} \times \varphi(V,a)$. In particular, the dimension of $\pi(V)$ is zero, which we assumed is not the case. Hence, $\dim(\varphi(V,a)) \le \dim(V)-1$ and
we have by induction that
\begin{equation}\label{eq-8}
|\varphi(V,a) \cap U^{n/m-1}| \le d^{\dim(V)-1} \cdot (\prod_{i=1}^{k} d_i)^{\dim(V)-1}.
\end{equation}
The bound on $|V \cap U^{n/m}|$ now follows immediately from~\eqref{eq:induction_size}, \eqref{eq-7} and \eqref{eq-8}.

\section{Variety Evasive Sets in Finite Fields}\label{sec-finite}

Using the construction given in Section~\ref{sec-algclosure}, Theorem~\ref{thm-variety-basic} and Theorem~\ref{thm-variety-buckets}, we can construct large finite sets in $\F^n$, where $\F$ is a finite field, that have small intersections with any variety of bounded dimension and degree (where now we think of the variety as a subset of $\F^n$). Of course, this would follow by showing that the variety $\V(f_1,\ldots,f_k)$ defined over the algebraic closure of $\F$ has many points in $\F^n$. This argument is essentially identical to the one given in \cite{DvirLovett11} (the construction is the same, only with weaker constraints on the exponents $d_i$) and so we will only sketch it here. Another topic of interest in application is the explicitness of the finite set obtained in $\F^n$. There are several different notions of explicitness but the one obtained by our methods (as is the one in \cite{DvirLovett11}) satisfies a very strong definition of explicitness which we discuss below. We will only discuss the construction in Theorem~\ref{thm-variety-basic} since the extension to the `bucketing' construction of Theorem~\ref{thm-variety-buckets} follows easily.

Suppose $\F$ is of size $q$. Let $U = \V(f_1,\ldots,f_k) \subset \bF^n$ be the variety defined in Theorem~\ref{thm-variety-basic} and let $U' = U \cap \F^n$. The most direct way to obtain large size and explicitness is to pick the exponents $d_1,\ldots,d_k$ (or any other set of $k$ exponents) to be coprime to $q-1$. For a carefully chosen $\F$ this added requirements will not increase by much the total degree of the polynomials (see \cite{DvirLovett11} for some exact computations). This choice will guarantee that (a) $U'$ is large and (b) $U'$ is explicit. To see both, notice that for every fixing of the last $n-k$ variables (or indeed any other set of size $n-k$) it is trivial to compute the unique setting of the first $k$ variables so that the resulting point $x_1,\ldots,x_n$ is in $V$. This can be done by a single matrix inversion operation (over $\F$) and $k$ exponentiations. We use the fact that the map $x \mapsto x^{d_i}$ is invertible over $\F$ for all $i \in [k]$. This shows that, assuming the $d_i$'s are coprime to $q-1$, $V'$ has size $q^{n-k}$ and that there is an efficiently computable mapping $\phi : \F^{n-k} \mapsto U'$ that is one-to-one (and the inverse is also efficiently computable). There is also a way to argue about the size of $U'$ for general choice of exponents but this makes the explicitness of the construction less obvious (see \cite{DvirLovett11} for details).

We summarize the above argument in two immediate corollaries of Theorem~\ref{thm-variety-basic} and Theorem~\ref{thm-variety-buckets}.

\begin{cor}\label{cor-variety-basic}
Let $1 \le k \le n$ and $d \geq 1$ be integers and let $\F$ be a field. Let $d_1>d_2>\ldots>d_n > d$ be pairwise relatively prime integers, and assume that at least $k$ of $d_1,\ldots,d_n$ are co-prime to $|\F|-1$. Let $U \subset \bF^n$ be the variety defined
by Theorem~\ref{thm-variety-basic}, and let $U' = U \cap \F^n$. Then
$$
|U'|=|\F|^{n-k},
$$
and for every affine variety $V \subset \F^n$ of dimension $k$ and degree at most $d$,
$$
|V' \cap U| \le d \cdot \prod_{i=1}^{k} d_i.
$$
\end{cor}

\begin{cor}\label{cor-variety-buckets}
Let $k,d \geq 1$ be integers, $\eps>0$ and let $\F$ be a field. Let $m>k/\eps$ be an integer such that $m$ divides $n$.
Let $d_1>d_2>\ldots>d_m > d$ be pairwise relatively prime integers, and assume that at least $k$ of $d_1,\ldots,d_m$ are co-prime to $|\F|-1$. Let $U^{n/m} \subset \bF^n$ be the variety defined
by Theorem~\ref{thm-variety-buckets}, and let $U' = U^{n/m} \cap \F^n$. Then
$$
|U'| = |\F|^{n(1-k/m)} \ge |\F|^{(1-\eps) n},
$$
and for every affine variety $V \subset \F^n$ of dimension $k$ and degree at most $d$,
$$
|V' \cap U| \le d^{k+1} \cdot (\prod_{i=1}^{k} d_i)^k.
$$
\end{cor}

\section{Comparison With a Random Construction}
\label{sec-random}

We compare in this section the explicit results we obtained, with results than one can get from random constructions. In many scenarios random constructions obtain optimal or near optimal parameters, and these can be compared to the best results than one can obtain explicitly. For technical reasons, our discussion in this section will be restricted to varieties that are defined over $\F$. This is in contrast to our explicit construction that works also for varieteis defined over an extension of $\F$. The main technical difficulty with varieties defined not over $\F$ is in bounding their number (this number is finite since we are only interested in points in $\F^n$). 

We recall the parameters we obtained in Corollary~\ref{cor-variety-buckets}. Let $k$ denote the dimension, $d$ the degree and $n$ the number of variables, and let $\eps>0$ denote a small parameter. Choosing a finite field $\F$ appropriately, we gave an explicit construction of a subset $S \subset \F^n$ of size $|S| \ge |\F|^{(1-\eps) n}$ such that for any affine variety $V \subset \bF^n$ of degree $d$ and dimension $k$,
$$
|S \cap V| \le (d+k/\eps)^{O(k^2)}.
$$

We compare in this section what parameters can one achieve, if $S \subset \F^n$ is chosen randomly of size $|S|=|\F|^{(1-\eps)n}$.
We analyze this random construction when the dimension of the variety is small enough, $k \ll \eps n$. We note that our simple analysis for a random construction breaks when $k \approx n$, while our explicit construction still get bounds which are independent of the field size.

Let $\mathcal{V}_{n,d,k}$ denote the family of varieties in $\bF^n$ of degree $d$ and dimension $k$ that are defined over $\F$.

\begin{lem}\label{lemma-random-set}
Let $n,d,k \ge 1, \eps>0$ be parameters, and assume that $k \le \eps n/4$. Let $\F$ be a field large enough such that $d \le |\F|^k$. Let $S \subset \F^n$ be a random subset of size $|S| = |\F|^{(1-\eps)n}$. Then with high probability over the choice of $S$, for all varieties $V \in \mathcal{V}_{n,d,k}$
$$
|S \cap V| \le O\left(\frac{d}{\eps} \cdot {k+d+2 \choose k}\right).
$$
\end{lem}

First, we need a bound on the number of points in $\F^n$ in a variety $V \in \mathcal{V}_{n,d,k}$.

\begin{claim}\label{claim-size-variety}
Let $V \in \mathcal{V}_{n,d,k}$. Then $|V \cap \F^n| \le d \cdot |\F|^k$.
\end{claim}

\begin{proof}
We prove the claim by induction of the number of variables, degree and dimension.
It suffices to prove the claim for irreducible varieties, since if $V = \cup V_i$ is the decomposition
of $V$ to irreducible varieties, then $\deg(V)=\sum \deg(V_i)$ and $\dim(V_i) \le \dim(V)$. So,
we assume that $V$ is irreducible. Let $H_c:=(x_1=c)$ for $c \in \F$ be a family of hyperplanes.
If $V \subset H_c$ for some $c$ then the claim follows by induction on the number of variables.
Otherwise let $V_c:=V \cap H_c$. Then $V_c$ has dimension $k-1$ and degree $\le d$. Hence
$$
|V \cap \F^n| \le \sum_{c \in \F} |V_c \cap \F^n| \le |\F| \cdot d |\F|^{k-1} = d |\F|^k.
$$
\end{proof}

We next need a bound on the number of varieties in $\mathcal V_{n,d,k}$. Recall that this set contains only varieties defined over the finite field $\F$.

\begin{claim}\label{claim-num-varieties}
$\left| \mathcal{V}_{n,d,k}\right| \le |\F|^{n \cdot O\left(d {k+d+2 \choose k}\right)}$.
\end{claim}

\begin{proof}
We first argue about irreducible varieties. Let $V$ be an irreducible variety of degree $d$ and dimension $k$. Assume w.l.o.g that $x_1,\ldots,x_k$ are algebraically independent over $V$. We first claim that there exist polynomials $\{f_i(x_1,\ldots,x_k,x_{k+i}): 1 \le i \le n-k\}$ with coefficients in  $\F$ of degree $d$, such that $V$ is an irreducible component of $U=\{x \in \bF^n: f_1(x)=\ldots=f_{n-k}(x)=0\}$. To see this, note that by Lemma~\ref{lem-projection} we can take $f_i$ to be the polynomial defined by projection to the variables $x_1,\ldots,x_k,x_{k+i}$. Since we assumed $x_1,\ldots,x_k$ are algebraically independent over $V$, we must have that $f_i$ depends on $x_{k+i}$. Let us decompose
$
f_i(x_1,\ldots,x_k,x_{k+i}) = \sum_{j=1}^{d_i} f_{i,j}(x_1,\ldots,x_k) \cdot x_{k+i}^{j},
$
where $1 \le d_i \le d$ and $f_{i,d_i}$ is not identically zero. Let $g(x_1,\ldots,x_k) = \prod_{i=1}^n f_{i,d_i}(x_1,\ldots,x_k)$, and let $G=\{x \in \bF^n:g(x)=0\}$ be the hypersurface defined by $g$. We have by construction that $U \setminus G$ has dimension $k$. Moreover, we have that $V \setminus G$ has dimension $k$, since $V$ is irreducible and is not contained in $G$. Thus, there exists a Zariski open subset of $U$ of dimension $k$ which contains a Zariski open subset of $V$. Since $V$ is irreducible, this can only happend in $V$ is an irreducible component of $U$.

%

So, we obtain that there exist $(n-k)$ polynomials of degree $d$ in $k+1$ variables over $\F$, such that the variety that they define have an irreducible component equal to $V$. Hence the number of distinct possibilities for $V$ is bounded by
$$
{n \choose k} \left(|\F|^{{k+d+2 \choose k}}\right)^{(n-k)} \le n^k |\F|^{n \cdot {k+d+2 \choose k}} \le |\F|^{n \cdot O\left({k+d+2 \choose k}\right)}.
$$
To get the bound for general, not necessarily irreducible varieties, we need to sum over all possible decompositions of $V$ into irreducible components of degree $d_1+\ldots+d_r=d$. Hence
\begin{eqnarray*}
| \mathcal{V}_{n,d,k}| &\le& \sum_{d_1+\ldots+d_r=d} \prod_{i=1}^r|\F|^{n \cdot O\left({k+d_i+2 \choose k}\right)}
\le |\F|^{n \cdot O\left(d {k+d+2 \choose k}\right)}.
\end{eqnarray*}
\end{proof}

We now prove Lemma~\ref{lemma-random-set}.

\begin{proof}[Proof of Lemma~\ref{lemma-random-set}]
Let $c>0$ be a parameter to be fixed later. Let $S \subset \F^n$ be a random subset of size $|S|=|\F|^{(1-\eps)n}$. We will show that with high probability over the choice of $S$, $|S \cap V| \le c$ for all $V \in \mathcal{V}_{n,d,k}$. In order to show this, consider first a fixed variety $V \in \mathcal{V}_{n,d,k}$. By Claim~\ref{claim-size-variety} we know that $|V \cap \F^n| \le d |\F|^k$, hence
$$
\Pr_S[|S \cap V| \ge c] \le {|V \cap \F^n| \choose c} |\F|^{-\eps n \cdot c} \le (d |\F|^{k - \eps n})^c \le |\F|^{-(\eps/2)n \cdot c},
$$
by our choice of parameters.
The number of distinct $V \in \mathcal{V}_{n,d,k}$  is bounded by Claim~\ref{claim-num-varieties} by at most $|\F|^{n s}$ where $s=O(d {k+d+2 \choose k})$. So, for $c \ge O(s/\eps)$ we get by the union bound that with high probability, $|S \cap V| \le c$ for all $V \in \mathcal{V}_{n,d,k}$.
\end{proof}

\section{Connection to a Conjecture of Griffiths and Harris}\label{sec-conjecture}

Here we consider how Theorem~\ref{thm-variety-basic} fits with various known results and conjectures
about sub-varieties of complete intersections.

A {\em hypersurface} of degree $d$ is a zero set of a polynomial
of degree $d$; these form a vector space $V_{n,d}$.
A claim holds for a {\em very general} hypersurface if it holds
whenever the polynomial is outside a countable union of Zariski closed
sub-varieties of  $V_{n,d}$.  According to a conjecture of \cite{g-h},
if $X_d\subset {\mathbb P}^n$ is a very general projective hypersurface
of sufficiently high degree then for every sub-variety
$Z\subset X$, the degree of $X$ divides the degree of $Z$. The conjecture does not specify `sufficiently high degree', but the only known counter examples have
$d\leq 2n-3$.

The conjecture is not known in general. For $n=3$ this is the
classical Noether--Lefschetz theorem. In higher dimensions
only much weaker divisibility results are known using the method of
\cite{k-tr} and only some of these have been worked out explicitly.
For instance, if  $X_d\subset {\mathbb P}^4$ is a very general hypersurface
and $d=p^3$ for a prime $p\geq 5$ then $p$
divides the degree of  every subvariety $Z\subset X_{d}$.

Note further that it is known that one definitely needs a
  countable union of Zariski closed
subvarieties of  $V_{n,d}$ for the conjecture to hold,
thus the general methods may not guarantee the existence of
 examples over countable fields.
For a complete treatment see \cite[Chap.III]{voi-book} and
\cite{voi} for further related results.

Let us now assume the above conjecture and see what
it would imply if we replace the construction of Theorem~\ref{thm-variety-basic} with general complete intersections of the same degrees.
Applying the conjecture to several hypersurfaces, we get that
 if $d_1,\dots, d_k $ are pairwise relatively prime and
$$
X_{d_1,\dots, d_k}:=X_{d_1}\cap \cdots \cap X_{d_k}\subset {\mathbb P}^n
$$
is a very general complete intersection
of sufficiently high degree then
$d_1\cdots d_k$ divides the degree of
 every subvariety $Z\subset X_{d_1,\dots, d_k}$.

Let now $Y\subset {\mathbb P}^n$ be any subvariety
of degree $<\min_i\{d_i\}$.  Consider the sequence of intersections
$$
Y\supset Y\cap X_{d_1}\supset Y\cap X_{d_1, d_2}\supset \cdots
\supset Y\cap X_{d_1,\dots, d_k}.
$$
If the dimension drops at each step then
$Y\cap X_{d_1,\dots, d_k} $ is zero dimensional.
Otherwise there is an index $i$  such that
$Y_i:=Y\cap X_{d_1,\dots, d_i}$  has dimension $k-i$ but
$X_{d_{i+1}}$ contains one of the irreducible components
of $Y_i$. We know that $\deg Y_i=\deg Y\cdot d_1 \cdots  d_i$
and the degree of every  irreducible component
of $Y_{ij}\subset Y_i$ is divisible by $d_1 \cdots  d_i$.
If $Y_{ij}\subset X_{d_{i+1}}$ then its degree is also
divisoble by $d_{i+1} $. Thus
$$
\deg Y\cdot d_1 \cdots  d_i=\deg Y_i\geq \deg Y_{ij}\geq
d_1 \cdots  d_i\cdot d_{i+1},
$$
a contradiction.

The bound $\deg Y<\min_i\{d_i\}$ is optimal as shown by
an intersection of $X_{d_i}$ with a linear space of dimension $k+1$.

Let us note finally that \cite{g-h} and related works
consider projective varieties while the setting considered in  Theorem~\ref{thm-variety-basic} is
affine. In fact, the projective closures of our constructions
are very degenerate: their intersection with the hyperplane at infinity
is a linear space (with high multiplicity).

\bibliographystyle{alpha}
\bibliography{laurent}

\end{document}